\documentclass[runningheads]{llncs}

\usepackage[T1]{fontenc}
\usepackage{graphicx}
\usepackage{amsmath}
\usepackage{amssymb}

\usepackage{amsthm}
\usepackage{algorithm}
\usepackage{algpseudocode}
\usepackage{comment}
\usepackage{booktabs}
\usepackage{multirow}
\usepackage{array}
\usepackage{color}
\usepackage[hidelinks]{hyperref}
\usepackage{tikz}
\usetikzlibrary{positioning,arrows.meta,shapes.geometric,fit,backgrounds}
\usepackage{pgfplots}
\pgfplotsset{compat=1.17}
\usepackage{subcaption}
\usepackage{url}

\begin{document}

\title{Beyond Best Response: Quantal Stackelberg Deception as Insurance Against Attacker Misspecification\thanks{DISTRIBUTION A: Approved for Public Release, Distribution is Unlimited}}
\titlerunning{Beyond Best Response}

\author{%
Asif Rahman\inst{1} \and
Md. Abu Sayed\inst{1} \and
Ahmed Ann Noor Ryen\inst{1} \and
Ahmed Hemida\inst{2} \and
Charles A. Kamhoua\inst{3} \and
Christopher Kiekintveld\inst{1}%
}

\authorrunning{A. Rahman et al.}

\institute{%
University of Texas at El Paso, El Paso, TX 79968, USA\\
\email{\{arahman3,msayed,aryen\}@miners.utep.edu, cdkiekintveld@utep.edu}
\and
Yokogawa Corporation of America, Sugar Land, TX, USA\\
\email{ahmed.hemida@yokogawa.com}
\and
DEVCOM Army Research Laboratory, MD 20783, USA\\
\email{charles.a.kamhoua.civ@army.mil}%
}

\maketitle

\begin{abstract}
Stackelberg Security Games (SSG) assume that an attacker observes the defender's strategy and chooses the target that maximizes their expected utility perfectly. In most realistic applications this is not plausible, and in the case of cyber deception (e.g., using decoys) the purpose of the game is to induce uncertainty and mistakes. Quantal response is a common way to represent noise and mistakes in decision-making; here it replaces perfect best-response with a logit choice with rationality parameter $\lambda$ and results in a generalized Quantal Stackelberg Equilibrium (QSE), which recovers the classical solution exactly as $\lambda \rightarrow \infty$. We conduct a deeper analysis of how QSE can function as a generalized form of insurance against a variety of forms of model specification error/uncertainty; our analysis shows that QSE provides a practical way to address the important role of tie-breaking rules and model uncertainty in SSG from both a theoretical and practical perspective. We conduct an empirical evaluation in a cybersecurity case study with two networks and real vulnerabilities drawn from CVE and scored using the Common Vulnerability Scoring System (CVSS). QSE beats Stackelberg in realized defender utility spanning 144 scenarios with specification errors and 25 parameter configurations, with gains of 46\% to 175\% showing a substantial advantage in a wide variety of realistic cases. 
\end{abstract}

\section{Introduction}

Cyber deception uses false information (e.g., decoys or "honeypots") to manipulate the observations and beliefs of cyber attackers. Game-theoretic approaches for optimizing cyber deception~\cite{kiekintveld2015game} are designed to induce uncertainty and mistakes in attackers, often relying on them to mistake something fake for something real. For the most part, these models build on the standard framework of Stackelberg Security Games (SSG)~\cite{tambe2011security,sinha2018stackelberg} which assumes that defenders will commit to optimal strategies and attackers will best-respond, breaking ties in favor of the defender. While this is a common approach, there are two major issues: (1) there is a foundational mismatch between assuming an attacker with perfect knowledge and rationality, when the point of deception is to induce uncertainty and mistakes, and (2) in reality it is nearly impossible for players to share an exact game model (perfectly known payoffs, signals, rationality, etc.), and equilibrium solutions are not robust to these errors. The importance of both model uncertainty and tie-breaking rules has been recognized since early studies of SSG~\cite{kiekintveld2011approx,pita2012robust,kiekintveld2013inteveral}, but many models still default to perfect rationality assumptions. 

We combine this early work on uncertainty in SSG with work drawing on models of human behavior that model mistakes using the framework of quantal choice~\cite{mcfadden1976}, which assumes that choices will follow a logistic response function biased towards higher expected values but also allowing for mistakes. In the context of SSG, this replaces the standard best-respond function for the attacker with a logistic one with a single rationality parameter $\lambda$ that captures the magnitude of mistakes. When the defender optimizes against this attacker we have a Quantal Stackelberg Equilibrium (QSE). When $\lambda \to \infty$ this model converges to the standard SSE solution, while when $\lambda \to 0$ it converges to a purely random attacker. While QR-based approaches have been studied in the context of human attackers~\cite{yang2012computing}, here we extend this to consider the value of this model as a more general form of "insurance" against uncertainty, model specification errors, and bounded rationality. 


The mismatch between assumptions of model and player perfection is particularly extreme in ambitious cyber deception models. For example, in coordinated multi-domain deception~\cite{sayed2026coordinated} we pair cyber honeypots with physical decoys so the two layers corroborate each other; this requires multiple layers of assumptions about the effectiveness and impact of deception. We use this as the basis for a case study using a realistic game model drawing on real networks, vulnerabilities, and values from open sources. We then investigate how QSE compares with standard SSE, focusing specifically on the robustness features of the model varying forms of uncertainty. We expected the boundedly rational model to win by spreading deception across more targets. It mostly does not. For moderate rationality the optimal strategies are nearly identical for QSE and SSE, but QSE is still much more robust.  
The advantage hides somewhere more subtle. Strong Stackelberg reasoning assumes an indifferent attacker breaks ties in the defender's favor, and a real network is full of ties. For SSE the three Log4Shell replicas and the two Ripple20 gateways sit at exactly equal attacker utility. The quantal attacker splits them instead of resolving the tie in the defender's favor, and that split costs the Stackelberg defender about half a utility point against every adversary we simulated. The same hedge survives an incorrect attacker model, an incorrect $\lambda$, Gaussian noise, or a worst-case tie-breaking rule. The QSE strategy provides a broad insurance against misspecification errors in realistic games, with very little premium. 

We build upon this insight with both a theoretical examination of the properties of QSE and a large empirical examination in a realistic cyber deception case study. We make three main contributions: 
\begingroup
\renewcommand\thefootnote{}%
\footnotetext{Code availability: \url{https://github.com/ucchol/qse-cyber-deception}}
\endgroup

\begin{itemize}
\item[\textbullet] A QSE formulation of coordinated cyber-physical deception that recovers the classical Stackelberg equilibrium exactly as $\lambda \to \infty$, with existence and recovery proof and an efficient gradient-based solution method. 

\item[\textbullet] Evidence on two networks with real vulnerabilities and Common Vulnerabilities and Exposures (CVE) identifiers, scored under CVSS v2 and v3, that QSE beats Stackelberg in realized utility across all cases for 144 mis-estimation scenarios and 25 parameter configurations. The QSE advantage is robust and ranges from 46\% to 175\%.

\item[\textbullet] A formal analysis in a booked versus realized theorem that traces the advantage to tie-breaking rather than coverage. Wherever a network holds interchangeable assets, perfect-rationality analysis overstates its security and any slightly random attacker collects the difference.
\end{itemize}

\section{Related Work and Research Questions}
\label{sec:related}

Stackelberg games became the standard model of defender-attacker interaction after deployments such as ARMOR at Los Angeles International Airport~\cite{pita2008armor}, with the defender committing first to a randomized allocation and the attacker observing it and playing a best-response. Conitzer and Sandholm~\cite{conitzer2006computing} established the computational foundations, Tambe~\cite{tambe2011security} wrote the canonical account, and Sinha et al.~\cite{sinha2018stackelberg} provide a long-term survey. The strong Stackelberg equilibrium (SSE) is a convention in much of the research, and assumes the follower breaks ties in the leader's favor, usually defended by noting that the leader can perturb its strategy infinitesimally and make the preferred response strictly optimal~\cite{vonstengel2010leadership,conitzer2016stackelberg}. That defense assumes a follower who best-responds, even as the variations disappear and many responses become equivalent. For cyber deception, the framework has been used for zero-day defense~\cite{sayed2022cyber}, and Asghar et al.~\cite{asghar2024scalable} built a double-oracle algorithm for scalable multi-domain deception. The coordinated cyber-physical model we extend is that of Sayed et al.~\cite{sayed2026coordinated}.

The perfect-rationality assumption has known limits. McKelvey and Palfrey~\cite{mckelvey1995qre} introduced the Quantal Response Equilibrium, in which agents choose actions with probability proportional to a logit function of expected utility. Yang et al.~\cite{yang2012computing} brought quantal response to security games with the PASAQ algorithm behind PROTECT, its first field use~\cite{shieh2012protect}. Nguyen et al.~\cite{nguyen2013analyzing} validated the approach in human-subject experiments. \v{C}ern\'{y} et al.~\cite{cerny2021qse} computed QSE in extensive-form games, proving NP-hardness and giving MILP approximations. A different approach avoids behavioral models altogether. Pita et al.~\cite{pita2012robust} proposed MATCH, which hedges against attacker deviations through robust optimization rather than modeling them, and can beat quantal methods when the behavioral parameters are miscalibrated~\cite{nguyen2013analyzing}. The distinction matters for us because MATCH bounds deviations where QSE prices them, and the deviations we exploit have structure worth pricing. 

Game-theoretic work on cyber-physical security has grown quickly. Guo et al.~\cite{guo2018smart} modeled smart grid security as a continuous Markov game without deception, and Pawlick, Colbert, and Zhu~\cite{pawlick2019taxonomy} gave a taxonomy of defensive deception spanning honeypots, moving target defense, and obfuscation. On honeypots specifically, P\'{i}bil et al.~\cite{pibil2012honeypot} cast placement as a network security game and Rahman et al.~\cite{rahman2025choice} used choice overload from behavioral economics as the deception itself. Chen and Li~\cite{chen2024} investigate logit-style tie-breaking rules in cases with multiple defenders. A gap remains, since cyber-physical deception models assume a perfectly rational attacker, and while quantal response has been successful in physical security applications~\cite{yang2012computing,shieh2012protect} it has not been evaluated in cyber-physical deception, where inducing attacker error is the whole point. We frame three specific research questions to advance the literature: 

\begin{description}
\item[RQ1.] Does optimizing against a bounded-rational attacker yield higher realized defender utility than against a perfectly rational one?
\item[RQ2.] Is that advantage robust to mis-estimating the attacker's rationality, and does it persist across game parameter configurations?
\item[RQ3.] Does the robustness extend beyond an incorrect rationality level to different incorrect models of the attacker's decision rule?
\end{description}

\section{Game Model and Solution Concepts}
\label{sec:method}

We consider a two-player interaction in a cyber-physical setting where the defender uses deception to protect the network. The players know the structure of the network and the vulnerabilities present (which are drawn from National Vulnerability Database (NVD) and scored using the Common Vulnerability Scoring System (CVSS)). The costs of decoys are also known, so the base game model is common knowledge (this changes when we consider potential error later). The defender does not directly observe the attacker's rationality, but has a prior belief over $\lambda$. The defender commits publicly to a randomized strategy, so the attacker sees the distribution of decoys but not the realization. The attacker cannot easily distinguish between real and decoy objects. The attacker will generally respond using a quantal response function rather than best response. We model a single interaction and the game is not zero-sum. 

\subsection{Game Model}

The network has $n$ nodes, partitioned into a cyber domain $\mathcal{C}$ and a physical domain $\mathcal{P}$. Vulnerabilities associated with resources in the network all have a Common Vulnerabilities and Exposures (CVE) identifier. The defender chooses a mixed strategy $\mathbf{A}_d \in \Delta^{n-1}$, assigning probability $A_d(i)$ to deploying deception on node $i$, and the attacker chooses a single target $j \in \{1, \ldots, n\}$ to attack.

When the defender places a honeypot on node $i$ and the attacker targets node $j$, the payoffs follow the CVSS-based structure of~\cite{sayed2026coordinated}. Let $V_d(i) \in [-10, -1]$ denote the defender's valuation of node $i$ (more negative means more valuable), and $V_a(i) \in [1, 10]$ the attacker's valuation. Let $c_d(i)$ be the defender's cost of deploying deception on node $i$, and $c_a(j)$ the attacker's cost of targeting node $j$. The reward matrices are
\begin{equation}
R_d(i, j) = \begin{cases}
\text{cap} \cdot V_a(i) - c_d(i) + c_a(j) & \text{if } i = j \text{ (capture)} \\
\text{esc} \cdot V_d(j) - c_d(i) + c_a(j) & \text{if } i \neq j \text{ (escape)}
\end{cases}
\label{eq:Rd}
\end{equation}
\begin{equation}
R_a(i, j) = \begin{cases}
-\text{cap} \cdot V_a(j) + c_d(i) - c_a(j) & \text{if } i = j \\
\;\;\,\text{esc} \cdot V_a(j) + c_d(i) - c_a(j) & \text{if } i \neq j
\end{cases}
\label{eq:Ra}
\end{equation}
where cap and esc are multipliers, following \cite{sayed2026coordinated}. The capture entry rewards the defender in proportion to $V_a$, the attacker's valuation of the trapped node deliberately, since a capture is worth more when the attacker is pursuing a bigger value. Since $i = j$ on the diagonal the same quantity can be written $V_a(i)$ or $V_a(j)$. Each player's expenditure enters the two matrices with opposite signs, the defender pays $c_d(i)$ and the attacker gains it (conversely for $c_a(j)$). Forcing the adversary to spend resources is a benefit for the defender, and for deception much of the defender's return is often increased attacker effort so the cost terms are transfers and cancel in the sum $R_d + R_a$. What makes the game non-zero-sum is the valuation structure, since the attacker's escape payoff follows its own $V_a(j)$ while the defender's escape loss follows $V_d(j)$.

The defender chooses a commitment to maximize the expected reward given the attacker's response, and the attacker scores each target by its expected utility under the commitment,
\begin{equation}
\max_{\mathbf{A}_d \in \Delta^{n-1}} \; \sum_{i,j} A_d(i)\, \sigma_a(j \mid \mathbf{A}_d)\, R_d(i,j), \qquad\quad U^a_j(\mathbf{A}_d) = \sum_i A_d(i)\, R_a(i, j).
\label{eq:objectives}
\end{equation}
A perfectly rational attacker selects $\arg\max_j U^a_j$, which yields the classical Stackelberg model, and Section~\ref{sec:qse} replaces that exact maximization with a quantal response $\sigma_a$.
Given a CVE with CVSS scores (exploit probability $\Pr$, impact score $\text{IS}$, base score $\text{BS}$, exploitability score $\text{ES}$), the raw valuations are
\begin{equation}
\tilde{V}_d(i) = -\Pr(i) \cdot \text{IS}(i) + (1 - \Pr(i)) \cdot r_i, \quad \tilde{V}_a(i) = \Pr(i) \cdot \text{BS}(i) - (1 - \Pr(i)) \cdot \text{ES}(i)
\label{eq:valuations}
\end{equation}
where $r_i = 5$ for cyber and $8$ for physical nodes to reflect the higher consequence of physical compromise; raw values are normalized to $[-10, -1]$ and $[1, 10]$.

We use the logit quantal response because the closed-form softmax admits analytical gradients and because it is standard in the security games literature~\cite{yang2012computing,shieh2012protect,nguyen2013analyzing}. The independence of irrelevant alternatives assumptions means that attackers evaluate targets independently, which nested logit~\cite{cerny2021qse} relaxes at computational cost. Section~\ref{sec:discussion} revisits this tradeoff.

\subsection{Quantal Stackelberg Equilibrium}
\label{sec:qse}

Given the commitment $\mathbf{A}_d$, the quantal response selects a target $j$ with a probabililty increasing in the expected utility $U^a_j$: 
\begin{equation}
\sigma_a(j \mid \mathbf{A}_d; \lambda) = \frac{\exp(\lambda \cdot U^a_j)}{\sum_k \exp(\lambda \cdot U^a_k)}
\label{eq:softmax}
\end{equation}
The rationality parameter $\lambda \geq 0$ controls how concentrated this distribution is, from uniform at $\lambda = 0$ to the best response $\arg\max_j U^a_j$ as $\lambda \to \infty$.
The defender's expected utility is given by: 
\begin{equation}
U_d(\mathbf{A}_d; \lambda) = \sum_{i,j} A_d(i) \, \sigma_a(j \mid \mathbf{A}_d; \lambda) \, R_d(i, j) = \mathbf{A}_d^\top R_d \, \boldsymbol{\sigma}_a
\label{eq:defender_utility}
\end{equation}

\begin{definition}[Strong Stackelberg Equilibrium]
\label{def:sse}
In SSE the attacker observes $\mathbf{A}_d$ and plays a best response, breaking ties in favor of the defender. A Strong Stackelberg Equilibrium (SSE) is a commitment $\mathbf{A}_d^{\text{SSE}}$ that is optimal against this attacker response~\cite{conitzer2006computing,vonstengel2010leadership}. We write $V_{\text{SSE}}$ for the defender's utility in SSE. The SSE is unique, and the time breaking rule ensures existence. 
\end{definition}

\begin{definition}[Quantal Stackelberg Equilibrium]
For a given $\lambda$ a mixed strategy $\mathbf{A}_d^* \in \Delta^{n-1}$ is a Quantal Stackelberg Equilibrium if it maximizes $U_d(\mathbf{A}_d; \lambda)$ when the attacker responds via (\ref{eq:softmax}). We write $V_{\text{QSE}}(\lambda)$ for the optimal value.
\end{definition}

\begin{proposition}[Existence]
\label{prop:exist}
For every $\lambda \geq 0$ a QSE exists without a tie-breaking rule. 
\end{proposition}

\begin{proof}
The quantal response (\ref{eq:softmax}) is single-valued and continuous in $\mathbf{A}_d$, so $U_d(\cdot\,; \lambda)$ is continuous on the compact simplex and attains its maximum, and uniqueness of the response removes the selection problem that forces the SSE convention in Definition~\ref{def:sse}.
\end{proof}

\begin{theorem}[Stackelberg Recovery, general form]
\label{thm:recovery}
Let $G$ be any finite two-player normal-form game with bounded payoffs in which a leader commits to a mixed strategy and a follower responds by the logit rule (\ref{eq:softmax}). Then
\begin{enumerate}
\item[(i)] $\limsup_{\lambda \to \infty} V_{\text{QSE}}(\lambda) \leq V_{\text{SSE}}$, and
\item[(ii)] if for every $\varepsilon > 0$ there exists a commitment with a unique follower best response whose value to the leader is at least $V_{\text{SSE}} - \varepsilon$, then $\lim_{\lambda \to \infty} V_{\text{QSE}}(\lambda) = V_{\text{SSE}}$.
\end{enumerate}
\end{theorem}

\begin{proof}
For (i), take $\lambda_m \to \infty$ and optimal commitments $\mathbf{x}_m$ with $U_d(\mathbf{x}_m; \lambda_m) \to \limsup_\lambda V_{\text{QSE}}(\lambda)$, and by compactness pass to a subsequence with $\mathbf{x}_m \to \mathbf{x}^\dagger$. Let $T$ be the follower's best responses to $\mathbf{x}^\dagger$. By continuity there is $\eta > 0$ such that for large $m$ every target outside $T$ sits at least $\eta/2$ below the follower maximum at $\mathbf{x}_m$, so its logit probability is at most $e^{-\lambda_m \eta/2}$ and vanishes. Every limit point of the response is supported on $T$, so the limit of $U_d(\mathbf{x}_m; \lambda_m)$ is a convex combination of defender payoffs over $T$ at $\mathbf{x}^\dagger$, at most the defender-favorable selection, at most $V_{\text{SSE}}$.
For (ii), fix $\varepsilon > 0$ and a commitment $\mathbf{x}_\varepsilon$ whose unique best response $j$ earns the leader at least $V_{\text{SSE}} - \varepsilon$. Uniqueness gives a gap $\eta_\varepsilon > 0$ to every other target, so under (\ref{eq:softmax}) the mass of $j$ is at most $(n-1) e^{-\lambda \eta_\varepsilon}$ and $U_d(\mathbf{x}_\varepsilon; \lambda)$ converges to the value of $\mathbf{x}_\varepsilon$ under response $j$. Hence $\liminf_\lambda V_{\text{QSE}}(\lambda) \geq V_{\text{SSE}} - \varepsilon$ for every $\varepsilon$, and (i) gives the limit.
\end{proof}

The hypothesis of (ii) is the standard inducibility property of leadership games; it holds for generic payoffs~\cite{vonstengel2010leadership} and Section~\ref{sec:results} verifies the convergence numerically in our games. The theorem is a statement about the general class of games since nothing in the proof uses the CVSS structure, the costs, or the two domains (only finiteness, bounded payoffs, and the logit form). QSE is therefore a strict generalization of the classical SSE. Two cautions frame how to read the limit. First, it concerns the optimal value, and the optimal quantal commitment at large $\lambda$ need not resemble the SSE commitment, a distinction Proposition~\ref{prop:booked} below turns into the central result of the paper. Second, the convergence is not monotone, as Figure~\ref{fig:lambda_sweep} later shows. At very low $\lambda$ the optimal value exceeds the SSE benchmark because a nearly random attacker is a weak opponent. It bottoms out near $\lambda = 0.5$, and only then does it climb to the Stackelberg limit, so the most damaging adversary in this family is a moderately confused one rather than a flawless one.

\subsection{Properties of the Objective}

Three further facts about the objective shape what follows, two for the solver and one for the results.

\begin{proposition}[Smoothness]
For any fixed $\lambda > 0$ the defender utility $U_d(\mathbf{A}_d; \lambda)$ is a continuously differentiable function of $\mathbf{A}_d$ on the interior of the simplex, with directional derivatives into the simplex at the boundary.
\end{proposition}

\begin{proof}[Sketch]
The map is a composition of the linear map $\mathbf{A}_d \mapsto \mathbf{A}_d^\top R_a$ with the smooth softmax, multiplied by terms linear in $\mathbf{A}_d$. Compositions and products of smooth maps are smooth.
\end{proof}

Smoothness is what enables a gradient-based solver, whereas the Stackelberg objective jumps discontinuously when the attacker's argmax switches, which is why standard formulations typically use integer programming.

\begin{remark}[Non-convexity]
$U_d(\mathbf{A}_d; \lambda)$ is not concave in general. The softmax introduces sign-indefinite curvature, so we treat the problem as potentially multi-modal and use multi-start optimization. In practice the landscape on our games proved benign. Across all experiments, 49 or 50 of 50 restarts converged, and repeated solves agreed on the optimum to within $10^{-2}$ utility. The isolated exceptions arise at moderate and large $\lambda$, where a single restart terminates cleanly by the optimizer's own criterion at a local optimum a few hundredths below the best, from an unremarkable interior draw, and the multi-start budget, costing under a second at $n = 16$, absorbs the effect.
\end{remark}

\begin{proposition}[Tie Splitting]
\label{prop:ties}
Fix any defender strategy $\mathbf{A}_d$ and any two targets $j, k$ with $|U^a_j - U^a_k| \leq \delta$. Under the quantal response (\ref{eq:softmax}),
\[
\frac{\sigma_a(j)}{\sigma_a(k)} = \exp\!\big(\lambda (U^a_j - U^a_k)\big) \in \big[e^{-\lambda\delta},\, e^{\lambda\delta}\big].
\]
\end{proposition}

\begin{proof}
Immediate from (\ref{eq:softmax}), since the normalizer cancels in the ratio.
\end{proof}

\begin{lemma}[Exact ties split exactly]
\label{lem:split}
Fix $\mathbf{A}_d$ and let $T$ be a set of targets whose attacker utilities $U^a_j$ are equal for all $j \in T$. Under (\ref{eq:softmax}) the attacker's probability is identical across $T$ at every $\lambda \geq 0$.
\end{lemma}

\begin{proof}
Equal utilities give equal numerators in (\ref{eq:softmax}).
\end{proof}

\begin{proposition}[Booked versus realized value on a tie boundary]
\label{prop:booked}
Let $\mathbf{x}^*$ be an SSE commitment and $T$ the set of attacker best responses to it, with defender payoffs $D_t = \sum_i x^*_i R_d(i,t)$ for $t \in T$. The SSE convention books the value $\max_{t \in T} D_t$ for $\mathbf{x}^*$, while the realized value of the same commitment against the quantal attacker satisfies
\[
\lim_{\lambda \to \infty} U_d(\mathbf{x}^*; \lambda) \;=\; \frac{1}{|T|} \sum_{t \in T} D_t .
\]
The convention therefore overstates the commitment's value by
\[
\Delta(\mathbf{x}^*) \;=\; \max_{t \in T} D_t \;-\; \frac{1}{|T|} \sum_{t \in T} D_t \;\geq\; 0,
\]
with equality only when the tied targets pay the defender identically.
\end{proposition}

\begin{proof}
Targets outside $T$ have attacker utility below the maximum by some $\eta > 0$, so their combined probability is at most $(n-1) e^{-\lambda \eta}$ and vanishes as $\lambda \to \infty$. Within $T$ the utilities are exactly equal, so by Lemma~\ref{lem:split} the surviving mass splits uniformly at every $\lambda$. The limit and the gap follow.
\end{proof}

Two facts make the proposition consequential. The first is that an optimal SSE has many ties by construction rather than by accident. The multiple-LP formulation of SSE maximizes defender utility subject to the chosen target being weakly best for the attacker. At an optimum some of those weak inequalities commonly bind. In our games they bind emphatically; the v3 optimum ties five targets whose defender payoffs span 1.3 utility points, the v2 optimum ties nine spanning $-63.1$ to $-91.7$. $\Delta$ evaluates to about half a point in one game and roughly seventeen in the other, the two gaps Section~\ref{sec:results} reports. The second fact is that the leader cannot perturb its way out. The usual argument holds that the leader can shift its strategy infinitesimally to make the preferred response strictly optimal~\cite{conitzer2016stackelberg}, but Proposition~\ref{prop:ties} prices that argument against a quantal follower, since opening a gap the attacker requires a utility change of order $1/\lambda$, a real strategic move with real costs elsewhere in the game. The tie the defender most wants to win resists steering in a more basic way, since coverage is what makes a replica unattractive, and shifting probability toward one to make the attacker prefer it pushes the attacker away instead. Theorem~\ref{thm:recovery} and Proposition~\ref{prop:booked} together resolve what the experiments observe. The SSE value remains approachable, and the quantal optimum at large $\lambda$ attains it by arranging its own tied set to be equal in defender payoff. The SSE strategy is what fails, by exactly $\Delta$ in the limit, and Section~\ref{sec:misest} shows the realized value already sitting at that limit for every true rationality at or above one.

\subsection{Gradient-Based Solver}

Since $\mathbf{A}_d$ appears both directly in (\ref{eq:defender_utility}) and indirectly through $\boldsymbol{\sigma}_a$, the gradient carries two terms:
\begin{equation}
\nabla_{\mathbf{A}_d} U_d = R_d \boldsymbol{\sigma}_a + \lambda \left[ R_a (\mathbf{g} \odot \boldsymbol{\sigma}_a) - (R_a \boldsymbol{\sigma}_a) U_d \right]
\label{eq:gradient}
\end{equation}
where $\mathbf{g} = \mathbf{A}_d^\top R_d$ and $\odot$ denotes the element-wise product, the first term is a direct effect of shifting deception and the second an indirect effect through the attacker's response. We solve each restart with sequential least squares programming (SLSQP)~\cite{kraft1988slsqp}, which handles the simplex constraints directly and consumes the analytical gradient, as Algorithm~\ref{alg:qse} summarizes.

\begin{algorithm}[t]
\caption{QSE Solver (Multi-Start SLSQP)}
\label{alg:qse}
\small
\begin{algorithmic}[1]
\Require Reward matrices $R_d, R_a \in \mathbb{R}^{n \times n}$, rationality $\lambda > 0$, restarts $K$
\Ensure Defender strategy $\mathbf{A}_d^*$, utility $U_d^*$
\State $U_d^* \gets -\infty$
\For{$k = 1$ to $K$}
    \State Sample $\mathbf{x}_0 \sim \text{Dirichlet}(\mathbf{1}_n)$ \Comment{Random start on simplex}
    \State $\mathbf{x}_k \gets \text{SLSQP}(f = -U_d, \nabla f = -\nabla U_d, \mathbf{x}_0, \text{simplex constraints})$
    \State Compute $U_d(\mathbf{x}_k; \lambda)$ via (\ref{eq:softmax}), (\ref{eq:defender_utility}); \textbf{if} $U_d(\mathbf{x}_k; \lambda) > U_d^*$ \textbf{then} $\mathbf{A}_d^* \gets \mathbf{x}_k$, $U_d^* \gets U_d(\mathbf{x}_k; \lambda)$
\EndFor
\State \Return $\mathbf{A}_d^*, U_d^*$
\end{algorithmic}
\end{algorithm}

We use $K = 50$ restarts and verify the gradient against finite differences along simplex-feasible directions to relative error below $10^{-5}$.
Dedicated QSE algorithms exist and apply to our setting. \v{C}ern\'{y} et al.~\cite{cerny2020dinkelbach} rewrite the QSE objective as a fraction and use a Dinkelbach transformation to reduce the computation (for any general-sum normal-form game) to parametric subproblems solved by gradient or mixed-integer methods with explicit error bounds. This is a principled route when scale grows or guarantees are required. At $n = 16$ the direct approach above is a simpler sufficient alternative, since fifty restarts on the smooth objective already agree to within $10^{-2}$ in under a second.

\subsection{Complexity}
Each SLSQP iteration costs $O(n^2)$ for the matrix-vector products in (\ref{eq:softmax}), (\ref{eq:defender_utility}), and (\ref{eq:gradient}), so a full solve is $O(K n^2 T)$ with $T$ the per-restart iteration count, roughly $10^5$ products at $n = 16$ with $K = 50$, completing in under one second on the free Colab CPU tier. The SSE baseline is computed by the multiple-LP method of Conitzer and Sandholm~\cite{conitzer2006computing}: for each candidate attacker response $j$ we solve the defender's problem of maximizing $\mathbf{A}_d^\top R_d(\cdot, j)$ subject to $j$ being a best response, and keep the best of the $n$ solutions. This takes a few seconds at $n = 16$, and the full robustness matrix finishes in under a minute. Both methods are comfortably interactive at this scale, and a systematic scaling study is left to future work.

\section{Results and Analysis}
\label{sec:results}

\subsection{Experimental Setup}
\label{sec:setup}

Our experiments use the 8-node network from~\cite{sayed2026coordinated} and a 16-node extension. The 8-node network has five cyber vulnerabilities (Log4Shell, Ripple20, Heartbleed, SharePoint, Exchange) and three physical ones (a Siemens PLC, the VeraEdge sensor platform, and the Shekar endoscope camera). Log4Shell, a critical remote code execution flaw in the widely deployed Log4j logging library, and Ripple20, a family of flaws in an embedded TCP/IP stack on millions of Internet of Things devices, recur throughout the paper. The 16-node network adds redundant replicas of these systems, primary and backup servers, paired gateways, and twin controllers. To reflect real deployments the replicas are what create the near-tied target sets that Proposition~\ref{prop:ties} analyzes. Table~\ref{tab:cves} lists the CVEs.

\begin{table}[t]
\centering
\scriptsize
\caption{Real CVE vulnerabilities with CVSS v3 scores and each CVE's replica count in the 16-node network.}
\label{tab:cves}
\begin{tabular}{llcccc}
\toprule
CVE ID & Domain / System & Base Score & Impact & Exploit Prob. & Nodes \\
\midrule
CVE-2021-44228 & Cyber / Log4Shell      & 10.0 & 6.0 & 0.73 & 3 \\
CVE-2020-11898 & Cyber / Ripple20       & 9.1  & 5.2 & 0.73 & 2 \\
CVE-2014-0160  & Cyber / Heartbleed     & 7.5  & 3.6 & 0.73 & 2 \\
CVE-2020-16952 & Cyber / SharePoint     & 8.6  & 4.7 & 0.47 & 1 \\
CVE-2020-0688  & Cyber / Exchange       & 8.8  & 5.9 & 0.53 & 2 \\
CVE-2019-10915 & Physical / PLC         & 8.0  & 5.9 & 0.47 & 2 \\
CVE-2017-9389  & Physical / VeraEdge    & 8.8  & 5.9 & 0.53 & 2 \\
CVE-2017-10724 & Physical / Shekar endoscope & 8.8  & 5.9 & 0.53 & 2 \\
\bottomrule
\end{tabular}
\end{table}

Unless specified, $\text{cap} = 5$ and $\text{esc} = 10$. Costs vary by node role around the flat values of~\cite{sayed2026coordinated}, cyber nodes have $c_d \in [8, 15]$ and $c_a \in [2, 5]$ while physical nodes have $c_d \in [4, 7]$ and $c_a \in [8, 13]$, since physical decoys are cheaper to field when paired with cyber replicas while physical attacks remain costly.

To make the payoff structure concrete, consider the primary Log4Shell server, node 1, which under CVSS v3 has base score 10.0, impact score 6.0, exploit probability 0.73, and exploitability score 3.89. Equation (\ref{eq:valuations}) gives raw valuations $-0.73 \times 6.0 + 0.27 \times 5 = -3.03$ and $0.73 \times 10.0 - 0.27 \times 3.89 = 6.25$, which normalize across the network to $V_d(1) = -10$ and $V_a(1) = 10$ (the most valuable target for both players). With $c_d(1) = 15$, since the cross-domain replica that coordinates deception pairs with this server adds overhead, and $c_a(1) = 3$, a capture on node 1 pays $R_d(1,1) = 5 \times 10 - 15 + 3 = 38$, while an escape to the core Ripple20 gateway with $V_d(4) = v$ pays $R_d(1,4) = 10v - 12$. Escapes cost roughly three times what captures earn, which is why any sensible defender randomizes over several nodes The experiments bear this out, since the Greedy strategy that concentrates all deception on Log4Shell performs worse than uniform random.

We compare QSE against three baselines. URS places deception uniformly at random. Greedy places all deception on the single highest-value target. SSE is the strong Stackelberg equilibrium computed by the multiple-LP method above, validated by reproducing the defender utilities of~\cite{sayed2026coordinated} on the 8-node game exactly. For the convergence analysis we sweep $\lambda$ from 0.1 to 500. For the misestimation analysis we vary the defender's assumed value $\lambda^*$ and the attacker's true value $\lambda$ independently over $\{0.5, 1, 2, 5, 10, 50\}$, giving a $6 \times 6$ matrix of realized utilities. The parameter-sensitivity sweep rebuilds the reward matrices for every configuration with its own normalization, so utilities compare within a configuration only, and gaps there are reported relative to realized SSE utility.

\subsection{Baselines Across Four Games}

Table~\ref{tab:comparison} compares the four solution concepts for both network sizes and both CVSS versions. Two patterns hold everywhere. First, Greedy is the worst strategy in every game, worse even than uniform random, realizing $-106.53$ against URS's $-87.00$ on the 8-node CVSS v3 game, because a strategic attacker simply evades concentrated deception while the defender pays the full deployment cost. Second, SSE and QSE are close together when each is evaluated for its own model assumptions, with QSE at $\lambda=2$ trailing SSE by 0.4 to 2.4 utility points across the four games. That nominal gap is the premium QSE pays for hedging, and the robustness analysis below shows what the premium buys.

The strategy vectors explain where the advantage comes from. Under CVSS v3 the SSE and QSE allocations are nearly indistinguishable, both spread across the five Log4Shell and Ripple20 replicas and nothing else, with every weight within half a point. The advantage in that game comes entirely from how the weights interact with tie splitting, rather than from covering different nodes. Under CVSS v2 the compressed valuations make secondary targets competitive and both strategies spread widely--across nine nodes for SSE and eight for QSE. These include the Exchange servers and a physical VeraEdge sensor at roughly two percent showing quantitatively how cross-domain deception enters the portfolio exactly when no single domain dominates. We had expected QSE to spread strictly wider than SSE at moderate $\lambda$, but the supports largely coincide. The advantage flows through weight calibration and tie handling rather than extra coverage, and the spreading intuition survives only at very low rationality. QSE wins on pricing, not on placement.

\begin{table}[t]
\centering
\scriptsize
\caption{Defender utility across solution concepts with QSE at moderate $\lambda = 2$.}
\label{tab:comparison}
\begin{tabular}{lrrrr}
\toprule
 & \multicolumn{2}{c}{8-node} & \multicolumn{2}{c}{16-node} \\
\cmidrule(lr){2-3} \cmidrule(lr){4-5}
Method & CVSS v2 & CVSS v3 & CVSS v2 & CVSS v3 \\
\midrule
URS              & $-75.52$  & $-87.00$  & $-84.45$  & $-97.75$ \\
Greedy           & $-88.88$  & $-106.53$ & $-92.88$  & $-107.00$ \\
SSE              & $-52.10$  & $-48.94$  & $-63.09$  & $-76.92$ \\
QSE ($\lambda=2$) & $-54.23$  & $-49.88$  & $-65.51$  & $-77.28$ \\
\bottomrule
\end{tabular}
\end{table}

\subsection{Convergence to Stackelberg}

Figure~\ref{fig:lambda_sweep} traces QSE defender utility as $\lambda$ grows on the 8-node CVSS v3 game, confirming both predictions of Section~\ref{sec:method}. The curve converges to the SSE value, from $-49.88$ at $\lambda = 2$ to a gap of $0.02$ at $\lambda = 500$, and not monotonically, since at $\lambda = 0.1$ QSE realizes $-47.34$ against the SSE benchmark of $-48.94$ and the minimum of $-51.19$ is at $\lambda = 0.5$. Attack entropy falls from $1.85$ to under $0.01$ across the sweep, tracing the full stochastic-to-deterministic transition that a single Stackelberg point cannot represent.

\begin{figure}[t]
\centering
\begin{tikzpicture}
\begin{axis}[
    width=0.80\textwidth,
    height=4.7cm,
    xlabel={Attacker rationality $\lambda$},
    ylabel={Defender utility $U_d$},
    xmode=log,
    log basis x=10,
    xmin=0.08, xmax=650,
    ymin=-51.9, ymax=-46.7,
    ytick={-51,-50,-49,-48,-47},
    grid=major,
    grid style={gray!25},
    legend pos=south east,
    legend style={font=\footnotesize, row sep=-1pt, inner sep=2pt},
    label style={font=\small},
    tick label style={font=\small}
]
\addplot[color=teal, mark=*, mark size=1.4pt, thick] coordinates {
    (0.1, -47.34) (0.5, -51.19) (1, -50.55) (2, -49.88) (5, -49.48) (10, -49.28) (50, -49.04) (200, -48.97) (500, -48.96)
};
\addlegendentry{QSE}
\addplot[color=red, dashed, thick] coordinates {
    (0.08, -48.94) (650, -48.94)
};
\addlegendentry{SSE benchmark}
\end{axis}
\end{tikzpicture}
\caption{Defender utility against attacker rationality $\lambda$ on the 8-node CVSS v3 game, tracing the non-monotone path anticipated in Section~\ref{sec:method}. The axis is scaled to the QSE and SSE range, with uniform random far below at $-87.00$ (Table~\ref{tab:comparison}).}
\label{fig:lambda_sweep}
\end{figure}
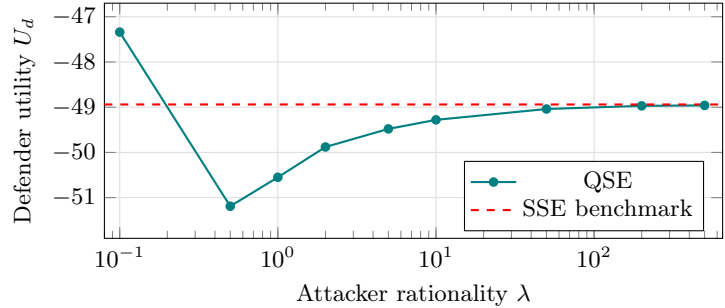

\subsection{Robustness to Estimation Error}
\label{sec:misest}

A central question is what happens when the defender's assumption about $\lambda$ is wrong. We computed the QSE strategy for assumed $\lambda^*$ in $\{0.5, 1, 2, 5, 10, 50\}$ and evaluated every strategy, plus SSE, against attackers for every true $\lambda$ in the same grid for the 16-node CVSS v3 game. Figure~\ref{fig:heatmap} shows the realized advantage of QSE over SSE in each case.

Repeating the grid for CVSS v2 and on both 8-node games gives the same verdict in every one of the remaining 108 cells, so QSE wins all 144 comparisons. What changes is the volume. Where the v3 tie spans defender payoffs only $1.3$ points, the v2 tie binds the web server the convention books at $-63.1$ to tied alternatives worth $-85.5$ to $-91.7$, including a physical sensor node, so realized utility lands near $-80$, a booked-versus-realized gap of roughly $17$ points. The tie tax in each game tracks one quantity, the spread of defender payoffs inside the tie set the SSE optimizer itself creates.

\begin{figure}[t]
\centering
\let\tiny\scriptsize
\scalebox{0.88}{\begin{tikzpicture}[scale=1.0]
\node[font=\footnotesize] at (-1.35, 3.75) {$\lambda^*_{\text{def}} \rightarrow$};
\node[font=\footnotesize] at (0.0, 3.75) {0.5};
\node[font=\footnotesize] at (0.95, 3.75) {1};
\node[font=\footnotesize] at (1.9, 3.75) {2};
\node[font=\footnotesize] at (2.85, 3.75) {5};
\node[font=\footnotesize] at (3.8, 3.75) {10};
\node[font=\footnotesize] at (4.75, 3.75) {50};
\node[font=\footnotesize, rotate=90] at (-2.0, 1.55) {$\lambda_{\text{atk}}$ (true)};
\node[font=\footnotesize] at (-1.1, 3.10) {0.5};
\node[font=\footnotesize] at (-1.1, 2.48) {1};
\node[font=\footnotesize] at (-1.1, 1.86) {2};
\node[font=\footnotesize] at (-1.1, 1.24) {5};
\node[font=\footnotesize] at (-1.1, 0.62) {10};
\node[font=\footnotesize] at (-1.1, 0.00) {50};
\fill[green!22!white, draw=gray!60] (-0.45,2.83) rectangle (0.45,3.37); \node[font=\tiny] at (0.0,3.10) {+0.05};
\fill[green!22!white, draw=gray!60] (0.50,2.83) rectangle (1.40,3.37); \node[font=\tiny] at (0.95,3.10) {+0.05};
\fill[green!31!white, draw=gray!60] (1.45,2.83) rectangle (2.35,3.37); \node[font=\tiny] at (1.90,3.10) {+0.12};
\fill[green!41!white, draw=gray!60] (2.40,2.83) rectangle (3.30,3.37); \node[font=\tiny] at (2.85,3.10) {+0.19};
\fill[green!38!white, draw=gray!60] (3.35,2.83) rectangle (4.25,3.37); \node[font=\tiny] at (3.80,3.10) {+0.17};
\fill[green!33!white, draw=gray!60] (4.30,2.83) rectangle (5.20,3.37); \node[font=\tiny] at (4.75,3.10) {+0.13};
\fill[green!27!white, draw=gray!60] (-0.45,2.21) rectangle (0.45,2.75); \node[font=\tiny] at (0.0,2.48) {+0.09};
\fill[green!29!white, draw=gray!60] (0.50,2.21) rectangle (1.40,2.75); \node[font=\tiny] at (0.95,2.48) {+0.10};
\fill[green!38!white, draw=gray!60] (1.45,2.21) rectangle (2.35,2.75); \node[font=\tiny] at (1.90,2.48) {+0.17};
\fill[green!51!white, draw=gray!60] (2.40,2.21) rectangle (3.30,2.75); \node[font=\tiny] at (2.85,2.48) {+0.26};
\fill[green!52!white, draw=gray!60] (3.35,2.21) rectangle (4.25,2.75); \node[font=\tiny] at (3.80,2.48) {+0.27};
\fill[green!47!white, draw=gray!60] (4.30,2.21) rectangle (5.20,2.75); \node[font=\tiny] at (4.75,2.48) {+0.23};
\fill[green!27!white, draw=gray!60] (-0.45,1.59) rectangle (0.45,2.13); \node[font=\tiny] at (0.0,1.86) {+0.09};
\fill[green!29!white, draw=gray!60] (0.50,1.59) rectangle (1.40,2.13); \node[font=\tiny] at (0.95,1.86) {+0.10};
\fill[green!38!white, draw=gray!60] (1.45,1.59) rectangle (2.35,2.13); \node[font=\tiny] at (1.90,1.86) {+0.17};
\fill[green!52!white, draw=gray!60] (2.40,1.59) rectangle (3.30,2.13); \node[font=\tiny] at (2.85,1.86) {+0.27};
\fill[green!57!white, draw=gray!60] (3.35,1.59) rectangle (4.25,2.13); \node[font=\tiny] at (3.80,1.86) {+0.30};
\fill[green!55!white, draw=gray!60] (4.30,1.59) rectangle (5.20,2.13); \node[font=\tiny] at (4.75,1.86) {+0.29};
\fill[green!26!white, draw=gray!60] (-0.45,0.97) rectangle (0.45,1.51); \node[font=\tiny] at (0.0,1.24) {+0.08};
\fill[green!27!white, draw=gray!60] (0.50,0.97) rectangle (1.40,1.51); \node[font=\tiny] at (0.95,1.24) {+0.09};
\fill[green!37!white, draw=gray!60] (1.45,0.97) rectangle (2.35,1.51); \node[font=\tiny] at (1.90,1.24) {+0.16};
\fill[green!55!white, draw=gray!60] (2.40,0.97) rectangle (3.30,1.51); \node[font=\tiny] at (2.85,1.24) {+0.29};
\fill[green!64!white, draw=gray!60] (3.35,0.97) rectangle (4.25,1.51); \node[font=\tiny] at (3.80,1.24) {+0.35};
\fill[green!65!white, draw=gray!60] (4.30,0.97) rectangle (5.20,1.51); \node[font=\tiny] at (4.75,1.24) {+0.36};
\fill[green!26!white, draw=gray!60] (-0.45,0.35) rectangle (0.45,0.89); \node[font=\tiny] at (0.0,0.62) {+0.08};
\fill[green!24!white, draw=gray!60] (0.50,0.35) rectangle (1.40,0.89); \node[font=\tiny] at (0.95,0.62) {+0.07};
\fill[green!33!white, draw=gray!60] (1.45,0.35) rectangle (2.35,0.89); \node[font=\tiny] at (1.90,0.62) {+0.13};
\fill[green!52!white, draw=gray!60] (2.40,0.35) rectangle (3.30,0.89); \node[font=\tiny] at (2.85,0.62) {+0.27};
\fill[green!68!white, draw=gray!60] (3.35,0.35) rectangle (4.25,0.89); \node[font=\tiny] at (3.80,0.62) {+0.38};
\fill[green!73!white, draw=gray!60] (4.30,0.35) rectangle (5.20,0.89); \node[font=\tiny] at (4.75,0.62) {+0.42};
\fill[green!23!white, draw=gray!60] (-0.45,-0.27) rectangle (0.45,0.27); \node[font=\tiny] at (0.0,0.00) {+0.06};
\fill[green!19!white, draw=gray!60] (0.50,-0.27) rectangle (1.40,0.27); \node[font=\tiny] at (0.95,0.00) {+0.03};
\fill[green!23!white, draw=gray!60] (1.45,-0.27) rectangle (2.35,0.27); \node[font=\tiny] at (1.90,0.00) {+0.06};
\fill[green!33!white, draw=gray!60] (2.40,-0.27) rectangle (3.30,0.27); \node[font=\tiny] at (2.85,0.00) {+0.13};
\fill[green!48!white, draw=gray!60] (3.35,-0.27) rectangle (4.25,0.27); \node[font=\tiny] at (3.80,0.00) {+0.24};
\fill[green!83!white, draw=gray!60] (4.30,-0.27) rectangle (5.20,0.27); \node[font=\tiny] at (4.75,0.00) {+0.49};
\node[font=\footnotesize, anchor=west, draw, rounded corners, fill=yellow!10, align=left] at (5.6, 1.55) {QSE wins\\\textbf{36/36} cells\\mean $+0.19$};
\end{tikzpicture}}
\caption{Realized advantage of QSE over SSE on the 16-node CVSS v3 game across assumed rationality $\lambda^*$ and true rationality $\lambda$. All 36 cells are positive, from $+0.03$ to $+0.49$ with mean $+0.19$, growing toward high rationality where SSE's reliance on favorable tie-breaking costs it most.}
\label{fig:heatmap}
\end{figure}

Every cell is positive. The advantages are modest in absolute terms, $+0.03$ to $+0.49$ utility points, but they are systematic and they are exactly what Proposition~\ref{prop:ties} predicts. The SSE strategy's model value for this game is $-76.92$, computed under the convention that the attacker breaks ties in the defender's favor. Against a quantal attacker at any rationality level, SSE instead realizes between $-77.41$ and $-77.45$, a consistent loss of roughly half a utility point. The source of the loss is visible in the strategy itself. The SSE optimizer does what classical security-game solutions do: it equalizes the attacker's utility across its attack set, leaving the three Log4Shell replicas and both Ripple20 gateways tied exactly. The convention then books the value of a favorable member, $-76.92$. But the tie is five wide, and by Proposition~\ref{prop:ties} a quantal attacker splits it evenly at every $\lambda$, putting an empirical $20\%$ on each of the five at $\lambda = 2$, including the two Ripple20 gateways worth $-78.2$ to the defender. SSE booked the favorable resolution as value. QSE priced the split during optimization and never owed the difference. The tie tax alone accounts for every cell of Figure~\ref{fig:heatmap}.

This reframes the misestimation question. The defender does not need an accurate estimate of $\lambda$ for QSE to pay off, because the dominant failure mode of SSE, optimism about tie resolution, is independent of which $\lambda$ the defender assumes. Even the worst example clears SSE by $+0.03$.

\subsection{Robustness to Attacker Model Misspecification}
\label{sec:offmodel}
 
The misestimation grid stresses one assumption while protecting another, since both the assumed and the true attacker stay inside the logit family. The harder question, the one the insurance framing rides on, is what happens when the attacker's decision rule is not logit at all. A defender who tunes a single $\lambda$ has committed to a particular shape of irrationality. If that shape is wrong, does the strategy still improve?
 
We fix the QSE strategy at a single operating point, $\mathbf{A}_d^{\text{QSE}}(\lambda=2)$, never retuning it, and measure its realized utility against a battery of attackers that are not quantal. Scored uses one consistent metric, the realized defender utility of a fixed strategy against a fixed attacker response,
\begin{equation}
U_d^{\mathrm{real}}(\mathbf{A}_d; \boldsymbol{\sigma}) = \mathbf{A}_d^\top R_d\, \boldsymbol{\sigma},
\label{eq:realized}
\end{equation}
so that QSE, SSE, and the baselines are scored on the same footing. Here the alternative model is the ground truth.
 
Each attacker is a response map $\boldsymbol{\sigma}(\mathbf{A}_d)$ that we substitute into (\ref{eq:realized}). Beyond the logit family of (\ref{eq:softmax}), we test four structurally different rules. The $\varepsilon$-greedy attacker plays a perfect best response with probability $1-\varepsilon$ and uniformly at random otherwise,
\begin{equation}
\sigma^{\varepsilon}_a(j) = (1-\varepsilon)\,\mathbf{1}\!\left[j = \arg\max_k U^a_k\right] + \frac{\varepsilon}{n}.
\end{equation}
The satisficing attacker is indifferent across a $\delta$-optimal attack set $\mathcal{B}_\delta = \{\, j : U^a_j \ge \max_k U^a_k - \delta \,\}$, which it resolves either uniformly, $\sigma^{\mathrm{sat}}_a(j) = \mathbf{1}[j \in \mathcal{B}_\delta]/|\mathcal{B}_\delta|$, or adversarially against the defender,
\begin{equation}
\sigma^{\mathrm{adv}}_a = \mathbf{e}_{j^\star}, \qquad j^\star = \arg\min_{j \in \mathcal{B}_\delta} \mathbf{A}_d^\top R_d(\cdot, j).
\end{equation}
The adversarial variant is the pessimistic tie resolution that robust methods such as MATCH~\cite{pita2012robust} are built to withstand. The Gaussian, or probit, attacker draws additive noise of a different family,
\begin{equation}
\sigma^{\mathcal{N}}_a(j) = \Pr\!\left[ U^a_j + \eta_j \ge U^a_k + \eta_k \ \forall k \right], \qquad \boldsymbol{\eta} \sim \mathcal{N}(0, s^2 I),
\end{equation}
where we set $s = \pi / (\lambda\sqrt{6})$, the standard deviation of the Gumbel noise implicit in the logit response at rationality $\lambda$, so that logit and probit differ only in the \emph{shape} of the error and not its magnitude. Finally the level-1 attacker takes a single step of reasoning, best responding to a uniform defender, $\sigma^{L1}_a = \mathbf{e}_{j_0}$ with $j_0 = \arg\max_j (\bar{\mathbf{u}}^\top R_a)_j$ and $\bar{\mathbf{u}} = \mathbf{1}/n$.
 
Table~\ref{tab:offmodel} reports realized utilities for the 16-node CVSS v3 game and Figure~\ref{fig:offmodel} plots the full battery. The single QSE strategy, tuned only to logit at $\lambda=2$, strictly beats SSE against the entire logit family, every probit attacker, and both satisficers, twelve of the eighteen attackers in all, the remainder being four narrow losses and two effective ties at the uniform and lowest-rationality limits. The probit columns carry the reframing: a defender who assumes logit noise still outperforms perfect-rationality planning when the truth is Gaussian noise, so the value of QSE does not depend on the logit model being literally correct, only on the error being smooth and spread. The adversarial column is the worst-case tie resolution, and QSE beating it by as much as $+1.08$ makes a concrete argument for the robustness claim that MATCH-style methods make abstractly.
 
The strategy loses only to the two attackers that place a hard point mass on the single best response, $\varepsilon$-greedy and level-1, and only by $0.12$ to $0.24$, which is less a failure of the model than the boundary of its advantage: against a clost ot perfect best response consistent with SSE, QSE's hedging is mildly wasteful.
 
\begin{table}[t]
\centering
\scriptsize
\caption{Realized utility (\ref{eq:realized}) of a fixed QSE strategy ($\lambda=2$) against true attackers on the 16-node CVSS v3 game, with $\Delta$ computed from full-precision values. This section uses a different node-cost calibration than Table~\ref{tab:comparison}, and SSE books $-77.55$ here under favorable tie-breaking.}
\label{tab:offmodel}
\begin{tabular}{llrrrr}
\toprule
Family & Parameter & QSE $\lambda{=}2$ & SSE & URS & $\Delta$ \\
\midrule
\multirow{3}{*}{Quantal (logit)} & $\lambda=0.5$ & $-78.00$ & $-78.04$ & $-96.34$ & $+0.04$ \\
 & $\lambda=2$  & $-77.91$ & $-78.08$ & $-97.03$ & $+0.16$ \\
 & $\lambda=50$ & $-77.78$ & $-78.08$ & $-97.19$ & $+0.30$ \\
\midrule
\multirow{3}{*}{$\varepsilon$-greedy} & $\varepsilon=0.10$ & $-75.27$ & $-75.05$ & $-92.86$ & $-0.22$ \\
 & $\varepsilon=0.25$ & $-71.49$ & $-71.31$ & $-86.38$ & $-0.18$ \\
 & $\varepsilon=0.50$ & $-65.19$ & $-65.07$ & $-75.56$ & $-0.12$ \\
\midrule
\multirow{2}{*}{Satisficing} & uniform $\delta{=}0.5$ & $-77.77$ & $-78.08$ & $-97.19$ & $+0.30$ \\
 & adversarial $\delta{=}0.5$ & $-77.79$ & $-78.87$ & $-97.19$ & $+1.08$ \\
\midrule
\multirow{3}{*}{Probit (Gaussian)} & $\lambda=1$ & $-77.97$ & $-78.08$ & $-96.74$ & $+0.11$ \\
 & $\lambda=2$ & $-77.89$ & $-78.08$ & $-97.04$ & $+0.19$ \\
 & $\lambda=5$ & $-77.79$ & $-78.07$ & $-97.18$ & $+0.29$ \\
\midrule
Level-1 & --- & $-77.79$ & $-77.55$ & $-97.19$ & $-0.24$ \\
Uniform & --- & $-52.58$ & $-52.58$ & $-53.94$ & $-0.00$ \\
\bottomrule
\end{tabular}
\end{table}
 
\begin{figure}[t]
\centering
\includegraphics[width=\linewidth,height=4.1cm]{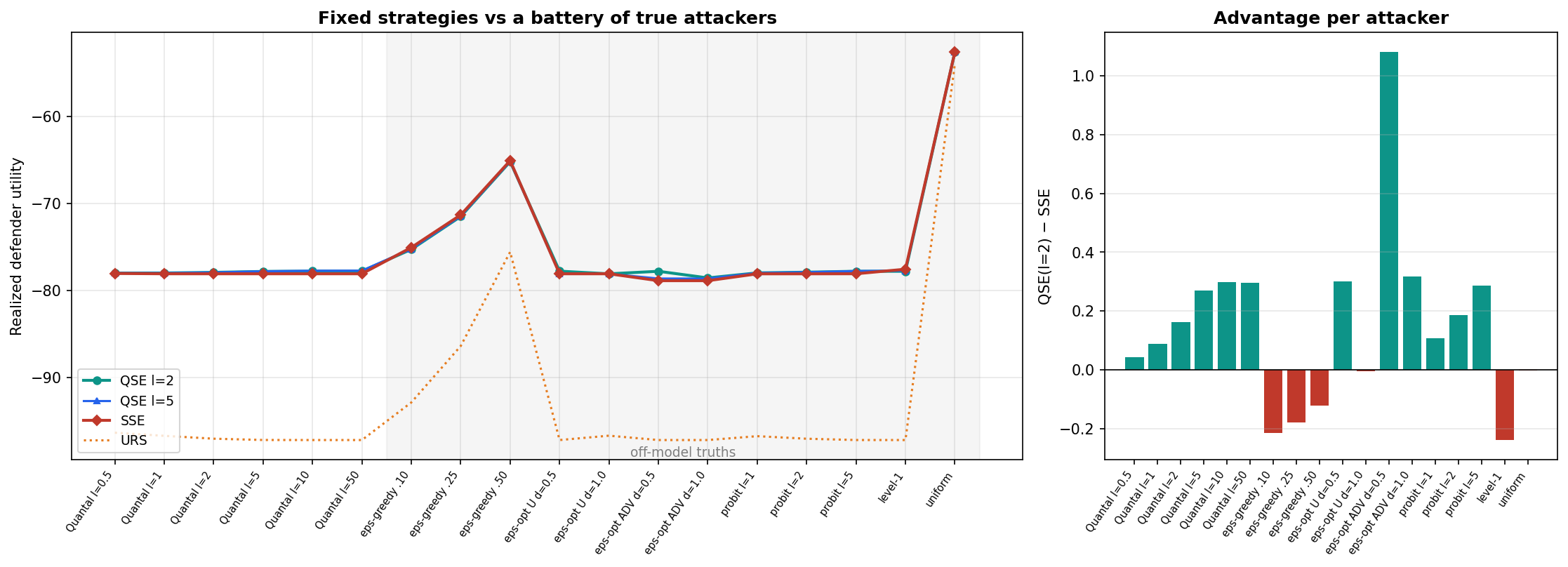}
\caption{Fixed strategies against the off-model attacker battery (16-node, CVSS v3), with the shaded region marking attackers outside the logit family and, right, the per-attacker advantage of QSE($\lambda=2$) over SSE.}
\label{fig:offmodel}
\end{figure}
 
The split traces to a structural fact already noted in Section~\ref{sec:misest}: the QSE and SSE strategies have nearly the same support. Since they deploy almost the same nodes, no attacker can separate them by much, which is why the realized advantage is small and, more importantly, why it is unconditional rather than tied to any property of the attacker. The mechanism is visible in SSE's realized utility, which among the tie-sensitive attackers, logit, probit, satisficing, and level-1, ranges from $-77.55$ to $-78.87$ in this section's calibration. SSE realizes its booked value of $-77.55$ exactly when the attacker cooperates with the favorable tie-breaking convention, as the level-1 attacker does by committing to the assumed best response. It realizes as little as $-78.87$ when the attacker resolves the tie adversarially. That spread is the tie tax of Section~\ref{sec:misest} measured directly, and QSE, which sits near $-77.8$ throughout, is precisely the insurance against it. The relationship is categorical rather than continuous in any scalar attacker property, separating attackers that concentrate on the single best response from those that distribute or mis-resolve ties.
 
To confirm this is not an artifact of one parameter setting, we rebuilt the game and reran the battery across eight configurations spanning the capture reward, the escape penalty, and an overall cost scale, with one representative attacker per mechanism. The sign pattern is invariant. The generalization and tie-break advantages stay positive in every configuration, growing with the capture reward to $+0.6$ against probit and $+2.9$ against the adversarial satisficer at cap $= 10$, the boundary losses to $\varepsilon$-greedy and level-1 stay mildly negative in every configuration, and the magnitudes scale exactly as the on-model sensitivity analysis of Section~\ref{sec:sensitivity} predicts.
 
The interpretation is the same as for the misestimation grid, extended one level. QSE does not merely tolerate a wrong estimate of the attacker's rationality; it tolerates a wrong model of the attacker's decision rule, across the smooth, satisficing, and adversarial families that describe a confused or noisy adversary, at the same near-zero premium. The single regime where perfect-rationality planning is preferable is the one it was designed for, a near-flawless best responder.

\subsection{Parameter Sensitivity}
\label{sec:sensitivity}

To check whether the robustness advantage is an artifact of one parameter configuration, we varied five parameters across five values each and re-solved the full game for all 25 configurations. For each configuration we computed the realized utilities of the SSE strategy and the QSE strategy, optimized at $\lambda^* = 2$, against attackers at $\lambda \in \{0.5, 2, 10, 50\}$, and we report the gap between mean realized utilities relative to SSE's, as described in Section~\ref{sec:method} summarized in Table~\ref{tab:sensitivity}.

\begin{table}[t]
\centering
\footnotesize
\caption{Parameter sensitivity on the 16-node game. The robustness gap is the improvement of mean realized QSE utility over mean realized SSE utility across attackers at $\lambda \in \{0.5, 2, 10, 50\}$. QSE wins all four checks in every configuration.}
\label{tab:sensitivity}
\begin{tabular}{lccc}
\toprule
Parameter varied & Values tested & QSE wins & Robustness gap \\
\midrule
Capture reward (cap)   & 2, 5, 10, 15, 20            & 5/5 & $+77\%$ to $+175\%$ \\
Escape penalty (esc)   & 3, 5, 10, 20, 30            & 5/5 & $+73\%$ to $+95\%$ \\
Cap:Esc ratio          & 2:13, 5:10, 7:8, 10:5, 13:2 & 5/5 & $+71\%$ to $+144\%$ \\
Cost scale             & $0.25\times$ to $4\times$    & 5/5 & $+84\%$ to $+88\%$ \\
Domain reward $(r_c, r_p)$ & $(0,0)$ to $(10,15)$    & 5/5 & $+46\%$ to $+85\%$ \\
\midrule
\textbf{Total} & \textbf{25 configurations} & \textbf{25/25} & $\mathbf{+46\%}$ \textbf{to} $\mathbf{+175\%}$ \\
\bottomrule
\end{tabular}
\end{table}

The robustness gap is positive in all 25 configurations, from $+46\%$ to $+175\%$, and QSE wins every individual rationality check. The gap grows monotonically with capture reward, from $+77\%$ at cap $= 2$ to $+175\%$ at cap $= 20$, since the more a capture is worth the more expensive it becomes to misprice which member of a tied set absorbs the attack. It stays between $+84\%$ and $+88\%$ across a sixteen-fold change in cost scale, so the advantage is structural rather than an artifact of calibration. The smallest gap occurs when the domain reward is zeroed out entirely, and even there QSE wins all four cases.

\subsection{Criticality and Calibration in Practice}

On the practical question of which vulnerabilities to prioritize, a leave-one-out criticality score (the change in defender utility when a node is removed) gives the same top five nodes in the same order under SSE and under QSE at $\lambda \in \{1, 5, 20\}$, the Log4Shell replicas and the two Ripple20 gateways, with criticality collapsing to near zero beyond them. This reassuring null result means the answer to which assets matter most does not depend on the rationality question, even though the optimal weights do.

The robustness analysis also removes most of the pressure to estimate $\lambda$ precisely, but a defender still has to pick a value. Three situations cover practice. If engagement data exists, $\lambda$ can be fit using maximum likelihood from observed attacker choices~\cite{yang2012computing}, and human-subject studies in security games~\cite{nguyen2013analyzing} typically find values between $0.5$ and $2$. If not, the threat profile is a reasonable guide: scanners and opportunists at low rationality and targeted operators at high. A defender who wants one number can take $\lambda^* \in [2, 10]$, the band where the realized advantage is consistently strong, with $\lambda^* = 5$ a sensible default. Even the worst choice in our grid still beats SSE, which is the main point. The cost of guessing wrong is small, but the cost of not modeling boundedness at all matters in every case.

A defender who distrusts any single estimate can optimize against the whole range instead, using the machinery already in hand. Let $\Lambda$ denote the grid $\{0.5, 1, 2, 5, 10, 50\}$ of the misestimation study. A Bayesian robust strategy maximizes the average of $U_d(\mathbf{A}_d; \lambda)$ under a uniform prior on $\Lambda$, a maximin strategy maximizes a smoothed minimum over the same grid, and both drop into Algorithm~\ref{alg:qse} unchanged since only the objective and its gradient differ. Table~\ref{tab:lambdaspace} reports how each commitment performs against every true rationality in $\Lambda$ on the 16-node CVSS v3 game.

\begin{table}[t]
\centering
\footnotesize
\caption{Realized defender utility against the true rationality for four commitments on the 16-node CVSS v3 game. Bayes averages the QSE objective uniformly over $\Lambda$, maximin optimizes a smoothed worst case, both via Algorithm~\ref{alg:qse}.}
\label{tab:lambdaspace}
\begin{tabular}{lcccccc}
\toprule
Commitment & $\lambda{=}0.5$ & $\lambda{=}2$ & $\lambda{=}10$ & $\lambda{=}50$ & Worst & Mean \\
\midrule
SSE                  & $-77.41$ & $-77.45$ & $-77.45$ & $-77.45$ & $-77.45$ & $-77.44$ \\
QSE ($\lambda^*{=}2$) & $-77.36$ & $-77.28$ & $-77.15$ & $-77.16$ & $-77.36$ & $-77.25$ \\
Bayes over $\Lambda$  & $-77.37$ & $-77.30$ & $-77.08$ & $-77.06$ & $-77.37$ & $-77.22$ \\
Maximin over $\Lambda$ & $-77.37$ & $-77.30$ & $-77.08$ & $-77.06$ & $-77.37$ & $-77.22$ \\
\bottomrule
\end{tabular}
\end{table}

The two robust strategies come out essentially identical They match the fixed $\lambda^* = 2$ commitment at the low end to within a few hundredths, and they improve on it by about a tenth of a point against the most rational attackers. The fixed moderate estimate therefore already captures most of the robustness available in this game, and the explicit treatment of $\Lambda$ earns its modest extra solve time when the plausible range is wider than ours or the tie structure richer. SSE is dominated by all three quantal commitments at every rationality in the grid.

\section{Discussion}
\label{sec:discussion}

Our theoretical and empirical analysis point to a common underlying structure of Stackelberg Security Games that has significant implications for understanding robustness to uncertainty and bounded rationality. The key feature is that estimation errors affect what happens among the set of tied targets in the solution, and only errors that affect the tie-breaking behavior here will significantly affect the final outcome of the game. The discontinuous maximization introduced in SSE makes that concept fundamentally problematic, but the smoothness of QSE allows finding solutions that only vary subtly from the SSE solutions in terms of allocating resources, but which make a dramatic difference in robustness. We can quantify this gap in our analysis, and it shows up consistently in the empirical findings for a realistic case study as well. 

The ties are not an artifact of our games, but rather a fundamental feature of the solutions found by SSE (a fact exploited in some specialized solution techniques~\cite{kiekintveld2013inteveral}). In real applications this is exacerbated by the fact that many real assets are close to indistinguishable in terms of value (e.g., multiple identical routers or cloned workstations). 
Section~\ref{sec:offmodel} shows how little this depends on the attacker's exact form. A defender who assumes logit noise still beats SSE when the truth is Gaussian, and QSE keeps its edge even against the adversarial tie resolution that robust methods such as MATCH~\cite{pita2012robust} are built to withstand. It is also worth noting that the quantal response model has a long history and substantial empirical support within the psychology and choice literature as a coarse approximation of choice distributions in humans, so it is likely to be a good approximation for realistic human errors in real strategic games as well. 
 
Our claims have some limitations. The game is one-shot and we do not model attacker learning. The logit response assumes errors independent across targets,and both networks derive from the same eight CVEs. The rationality values are borrowed from human-subject studies in physical security, and the absolute margins are modest. The main finding is the consistency across all 144 misestimation cases and 25 configurations, not the magnitude.
 
\section{Conclusion and Future Work}
\label{sec:conclusion}

Motivated by the observation that model uncertainty and challenges in defining exact behaviors for bounded opponents (including humans) are core challenges when deploying solutions like Stackelberg Security Games in realistic settings, we investigated how well the concept of Quantal Response can address these concerns. We replaced the perfect best response in coordinated cyber-physical deception to define QSE, and presented both a formal analysis and empirical results from a realistic case study that demonstrate that this simple change has a significant impact on the overall robustness of the solutions that is consistent across a wide variety of different manipulations. We also show that this change allows for simple gradient-based solution techniques that do not require integer programming but are fast and reliable. This builds on earlier work on robustness and tie-breaking rules, and deepens the understanding of the reasons for this result in SSG, as well as the practical implications and value of adopting this type of smoothed response function. One of the most promising aspects of this result is that the cost of added robustness is minimal, both in terms of the insurance premium and the computational cost. The gain in robustness is also durable, in that it applies across many different forms of errors and does not require precise identification or estimation of what may be causing deviations. 

For future work, nested logit is an immediate step, since replicas are naturally nested and their correlation is exactly what independence of irrelevant alternatives misses. Estimating $\lambda$ from real engagement data would ground the parameter in the cyber domain rather than borrowed physical-security behavior. Exploring wider bands and distributional priors that extend the robust variants, and sequential extensions with growing attacker rationality would be valuable, as would scaling studies on larger networks and testing results across more variations in the game assumptions and opponent models.

\begin{credits}
\subsubsection{\ackname}
Research was sponsored by the DEVCOM Army Research Laboratory and was accomplished under Cooperative Agreement Number W911NF-23-2-0012. The views and conclusions contained in this document are those of the authors and should not be interpreted as representing the official policies, either expressed or implied, of the Army Research Laboratory or the U.S. Government. The U.S. Government is authorized to reproduce and distribute reprints for Government purposes, notwithstanding any copyright notation herein.
\end{credits}


\end{document}